\documentclass[a4paper,10pt,oneside]{article}

\usepackage{amsmath}
\usepackage{amssymb}
\usepackage{amsthm}
\usepackage{enumerate}    
\usepackage{graphicx}
\usepackage{array}

\makeatletter
\let\@afterindentfalse\@afterindenttrue
\@afterindenttrue
\makeatother    
    
\newtheorem{definition}{Definition}
\newtheorem{theorem}{Theorem}
\newtheorem{lemma}{Lemma}
\newtheorem{corollary}{Corollary}

\newcommand*{\N}{\mathbb{N}}
\newcommand*{\R}{\mathbb{R}}
\newcommand*{\C}{\mathbb{C}}
\newcommand*{\K}{\mathbb{K}}
\newcommand*{\ci}{\mathop{\mathrm{i}}\nolimits} 
\newcommand*{\la}[1]{\lambda_{#1}} 
\newcommand*{\dint}{\,\mathrm{d}}  
\newcommand*{\sa}[2]{\mathcal{M}_{#1,#2}^\text{sa}}
\newcommand*{\DN}[1]{\mathcal{D}_{#1}} 
\newcommand*{\B}[2]{\mathcal{B}_{#1}\gz{#2}} 
\newcommand*{\BK}{\B{1}{\K^{2\times 2}}}
\newcommand*{\BR}{\B{1}{\R^{2\times 2}}}
\newcommand*{\BC}{\B{1}{\C^{2\times 2}}}
\newcommand*{\E}[1]{\mathcal{E}_{#1}}

\newcommand*{\Tr}{\mathop{\mathrm{Tr}}\nolimits}

\newcommand*{\Vol}{\mathop{\textrm{Vol}}\nolimits}
\newcommand*{\id}{\mathop{\textrm{id}}\nolimits}
\newcommand*{\sgn}{\mathop{\textrm{sgn}}\nolimits}
\newcommand*{\Lm}[1]{L_{#1}}
\newcommand*{\Rm}[1]{R_{#1}}
\newcommand*{\1}{\mathbf{1}}
\newcommand*{\abs}[1]{\left\vert#1\right\vert}
\newcommand*{\norm}[1]{\left\Vert#1\right\Vert}
\newcommand*{\hsnorm}[1]{\left\Vert #1\right\Vert_{\text{HS}}}
\newcommand*{\gz}[1]{\left(#1\right)}
\newcommand*{\sz}[1]{\left[#1\right]}
\newcommand*{\kz}[1]{\left\{#1\right\}}
\newcommand*{\rightabs}[1]{\left.#1\right\vert}
\newcommand*{\leftsz}[1]{\left[#1\right.}
\newcommand*{\rightsz}[1]{\left.#1\right]}
\newcommand*{\leftgz}[1]{\left(#1\right.}
\newcommand*{\rightgz}[1]{\left.#1\right)}
\newcommand*{\Li}{\mathop{\mathrm{Li}}\nolimits}

\title{
Invariance of separability probability over reduced states in $4\times 4$ 
  bipartite systems
\thanks{keywords: Entanglement, Peres-Horodecki condition, Separability 
 probability, Hilbert-Schmidt measure, Monotone metrics;
MSC: 81P16, 81P40, 81P45} }

\author{
Attila Lovas\thanks{lovas@math.bme.hu}, 
Attila Andai\thanks{andaia@math.bme.hu, Phone.: +36-1-4633127, 
  Fax.:+36-1-633172}, \\
Department for Mathematical Analysis, \\
Budapest University of Technology and Economics,\\
Stoczek u. 2, Budapest, H-1521, Hungary}

\date{\today}

\begin{document}

\maketitle

\begin{abstract}
The geometric separability probability of the composite quantum systems has 
  been extensively studied in the recent decades.
One of the simplest but strikingly difficult problem is to compute the 
  separability probability of qubit-qubit and
  rebit-rebit quantum states with respect to the Hilbert-Schmidt measure.
A lot of numerical simulations confirm the $P_{\mbox{\footnotesize 
  rebit-rebit}}=\frac{29}{64}$  and 
  $P_{\mbox{\footnotesize qubit-qubit}}=\frac{8}{33}$ conjectured probabilities.
Milz and Strunz studied the separability probability with respect to given 
  subsystems.
They conjectured that the separability probability of qubit-qubit (and  
  qubit-qutrit) states of the form of
  $\begin{pmatrix} D_{1} & C\\ C^{*} & D_{2} \end{pmatrix}$ depends on 
  $D=D_{1}+D_{2}$ (on single qubit subsystems), 
  moreover it depends only on the  Bloch radii ($r$) of $D$ and it is constant 
  in $r$.
Using the Peres-Horodecki criterion for separability we give mathematical 
  proof for the $\frac{29}{64}$ probability 
  and we present an integral formula for the complex case which hopefully will 
  help to prove the $\frac{8}{33}$ probability, too.
We prove Milz and Strunz's conjecture for rebit-rebit and qubit-qubit states.
The case, when the state space is endowed with the volume form generated by the 
  operator monotone function $f(x)=\sqrt{x}$ is also studied in detail.
We show that even in this setting the Milz and Strunz's conjecture holds true 
  and we give an integral formula for separability probability according to 
  this measure.
\end{abstract}

%81P40   Quantum coherence, entanglement, quantum correlations
%81P45   Quantum information, %communication, networks
%81P16   Quantum state spaces, operational and probabilistic concepts 

\section{Introduction}

Since entanglement is one of the most striking features of composite quantum 
  systems, it is natural to ask what the probability is
  that a given quantum state is entangled (or separable).
''Is the world more classical or more quantum? Does it contain more quantum 
  correlated (entangled) states than classically correlated ones?'' 
  These questions were addressed to physicists in 1998 by Zyczkowski, Horodecki,   
  Sanpera and Lewenstein \cite{ZHLM}.
The first question is a rather philosophical one, the second is easier to 
  formulate mathematically, although more specification is needed.
It has turned out during the recent years that even in the simplest quantum 
  case, when one considers only qubit-qubit states over real or complex 
  Hilbert-space, to determine the separability probability of a given state is a 
  highly nontrivial problem.
Many researchers agree and emphasize the philosophical and experimental interest 
  of the separability probability.
First, one should specify a natural measure on the state space and then should 
  compute somehow the volume of the separable states and the volume of the state 
  space.
In this paper we endow the state space with Hilbert-Schmidt measure which is 
  induced by the Hilbert-Schmidt metric.
We note here that other measures are also relevant, as it was pointed out by 
  Slater \cite{SlaterSensitivity},
  mainly those which are generated by monotone metrics \cite{PetzSud96}.
The volume of the state space with respect to the Hilbert-Schmidt measure was 
  computed by Zyczkowski \cite{ZyczMixedStates}  and Andai \cite{AndaiVol}.
There are several good separability criteria, we use the Peres--Horodecki 
  criterion \cite{Horodecki1} which is a simple necessary and sufficient 
  condition for separability of qubit-qubit states.
To compute the volume of separable states is a much more complicated task.

So far only extensive numerical studies and some related conjectures have 
  existed for the separability probability.
Numerical simulations give rise to an intriguing formula for separability 
  probability, presented in 2013 by Slater \cite{slater2013concise},
  which was tested in real, complex and even in quaternionic Hilbert-spaces 
  \cite{Slater2015,SlaterDunkl,FeiJoynt}.
Based on this formula and on numerical simulations the separability probability 
  for real qubit-qubit state is $\frac{29}{64}$ and for complex state is 
  $\frac{8}{33}$.
Now we give mathematical proof for the $\frac{29}{64}$ probability and we 
  present an integral formula for the complex case which hopefully will help to 
  prove the $\frac{8}{33}$ probability, too.
One of the most useful conjecture about separability probability was presented 
  by Milz and Strunz in 2015 \cite{Milz2015}.
They conjectured that the separability probability of qubit-qubit 
  (and qubit-qutrit) states of the form of
  $\begin{pmatrix} D_{1} & C\\ C^{*} & D_{2} \end{pmatrix}$ depends on 
  $D=D_{1}+D_{2}$ (on single qubit subsystems), 
  moreover it depends only on the  Bloch radii ($r$) of $D$ and it is constant 
  in $r$.
In this paper we prove this conjecture for real and complex qubit-qubit states.  

We study the case in detail when the state space is endowed with the volume form 
  generated by the operator monotone function $f(x)=\sqrt{x}$.
We show that the volume of rebit-rebit and qubit-qubit states are infinite, 
  although there is a simple and reasonable method to define the separability 
  probabilities.
We present integral formulas for separability probabilities in this setting, 
  too.
We argue that from the separability probability point of view, the main 
  difference between the Hilbert-Schmidt measure and the volume form generated 
  by the operator monotone function $x\mapsto\sqrt{x}$ is a special distribution
  on the unit ball in operator norm of $2\times 2$ matrices, more precisely in 
  the Hilbert-Schmidt case one faces with a uniform distribution on the whole 
  unit ball and for monotone volume forms one gets uniform distribution on the 
  surface of the unit ball.
\medskip

The paper is organized as follows.
In Section 2, we fix the notations for further computations and we mention some 
  elementary lemmas which will be used in the sequel.
In Section 3, we present our main results, namely an explicit integral formula 
  for the volume of separable qubit-qubit states over real and complex
  Hilbert-space, a proof for Milz and Strunz's conjecture \cite{Milz2015} and an 
  analytical proof for the rebit-rebit separability probability.
In Section 4 we endow the state space with the volume measure which induced by 
  the operator monotone function $f(x)=\sqrt{x}$, and we show, that even in this
  setting the Milz and Strunz's conjecture holds and we give an integral formula
  for separability probability according to this measure.
As a kind of checking, in Section 5, we compute the volume of the real and 
  complex qubit-qubit state space with methods introduced in Section 3
  and we compare our results to the previously published ones.
In the second part of Section 5 we prove that the volume of the qubit-qubit 
  state space is infinite if the volume measure comes from 
  the function $f(x)=\sqrt{x}$, but still there is a natural way to define the 
  separability probability.

\section{Basic lemmas and notations}

The quantum mechanical state space consists of real and complex self-adjoint 
  positive matrices with trace $1$. We consider only the set of faithful states 
  with real and complex entries.
In our notation the state space is
\begin{equation}
\DN{n,\K}= \{D\in\K^{n\times n}|D=D^\ast ,\,D>0,\,\Tr (D) =1\} \quad \K = \R,\C.
\end{equation}
The space of $n\times n$ self-adjoint matrices is denoted by $\sa{n}{\K}$ 
  ($\K = \R,\C$). 
Let us introduce the notation $\E{\K}$ for the operator interval
\begin{equation}
\E{n,\K} = \kz{\rightabs{Y\in\sa{n}{\K}} -I<Y<I}
\end{equation}
  where "$<$" denotes the partial ordering of self-adjoint matrices defined by 
  the cone of positive matrices.

The following lemma is an essential ingredient of the proof of our main theorem. 
It gives a characterization of positive definite matrices in terms of their 
  Schur complement.
\begin{lemma}\label{lem:schur}
For any symmetric matrix, D, of the form
\begin{equation*}
D = \gz{
\begin{array}{cc}
D_1    & C \\
C^\ast & D_2
\end{array}
},
\end{equation*}
  if $D_2$ is invertible then $D>0$ if and only if $D_2>0$ and 
  $D_1-CD_2^{-1}C^\ast>0$.
Similarly, if $D_1$ is invertible then $D>0$ if and only if $D_1>0$ and 
  $D_2-C^\ast D_1^{-1}C>0$.
\end{lemma}
\begin{proof}
The statement is well-known in linear algebra.
For the proof see for example p. 34 in \cite{zhang05}.
\end{proof}

% Singular value decomposition
\begin{lemma}\label{lem:svd}
For every matrix $V\in\K^{n\times n}$ there exists a factorization,
  called a singular value decomposition of the form
\begin{equation}
V = U_1\Sigma U_2,
\end{equation}
  where $U_1,U_2\in\K^{n\times n}$ are unitary matrices and 
  $\Sigma\in\K^{n\times n}$ is a diagonal matrix with real non-negative entries.
\end{lemma}
\begin{proof}
The proof can be found for example in Bathia's book 
  (See p. 6 in \cite{Bhatia}).
\end{proof}

% Norm and positivity
\begin{lemma}\label{lem:pn}
Let $X\in\K^{n\times n}$ be an arbitrary matrix. 
The matrix $X^\ast X$ is positive semidefinite and the following equivalence 
  holds
\begin{equation}
X^\ast X<I \Leftrightarrow \norm{X}<1,
\end{equation}
  where $\norm{\cdot}$ denotes the usual operator norm i.e. the largest singular 
  value  or Schatten-$\infty$ norm.
\end{lemma}
\begin{proof}
The inequality $\left<v,X^\ast X v\right>=\norm{Xv}^2\ge 0$ holds for all 
  $v\in\K^n$ which proves the first part of the statement.
By the definition of operator norm, we can write
\begin{equation*}
\norm{X}^2 
  =\sup\kz{\rightabs{\norm{Xv}^2}v\in\K^n,\,\norm{v}=1}\le 1
\end{equation*}
  because $\norm{Xv}^2=\left<v,X^\ast X v\right>\le \norm{v}^2=1$ for every 
  vector $v$ of length $1$.
\end{proof}

To a matrix $D\in\K^n$, one can associate the left and right multiplication 
  operators $\Lm{D},\Rm{D}:\K^{n\times n}\to\K^{n\times n}$ that acts like
\begin{align*}
A \mapsto \Lm{D} (A) &= DA \\
A \mapsto \Rm{D} (A) &= AD. 
\end{align*}
It is obvious that $\Lm{D}$ and $\Rm{D}$ are invertible if and only if 
  $D\in\text{Gl}(n,\K)$. By a straightforward computation, one can show that
\begin{equation*}
\det (\Lm{D}) = \det (\Rm{D}) = \det (D)^n.
\end{equation*} 
In integral transformations, the $n\times n$ complex matrix $D$ is regarded as a 
  $2n\times 2n$ real matrix that acts on $\R^{2n}\cong\C^n$ therefore the
  Jacobian of  $\Lm{D}$ and $\Rm{D}$ is $\det (D)^{2n}$ in the complex case.

The vector space of $n\times n$ matrices is the direct sum of  the space of 
  $n\times n$ self-adjoint matrices and $n\times n$ anti self-adjoint matrices 
\begin{equation*}
\K^{n\times n} =  \sa{n}{\K} \oplus \widetilde{\sa{n}{\K}}.
\end{equation*}
For any self-adjoint matrix $D\in\sa{n}{\K}$, the map 
  $\Lm{D}\circ\Rm{D}:\K^n\times\K^n$ preserves the direct sum decomposition i.e. 
  $\Lm{D}\circ\Rm{D} \gz{\sa{n}{\K}}\subset \sa{n}{\K}$ and 
  $\Lm{D}\circ\Rm{D} \gz{\widetilde{\sa{n}{\K}}}\subset   
  \widetilde{\sa{n}{\K}}$.
Consequently, 
  $\Lm{D}\circ\Rm{D}=\rightabs{\gz{\Lm{D}\circ\Rm{D}}}_{\sa{n}{\K}} \oplus
  \rightabs{\gz{\Lm{D}\circ\Rm{D}}}_{ \widetilde{\sa{n}{\K}}}$ holds which 
  implies that
\begin{equation}
\det\gz{\Lm{D}\circ\Rm{D}} =
\det\gz{\rightabs{\gz{\Lm{D}\circ\Rm{D}}}_{\sa{n}{\K}} } \times \det\gz{\rightabs{\gz{\Lm{D}\circ\Rm{D}}}_{\widetilde{\sa{n}{\K}}}}.
\end{equation}
This observation lead us to the following lemma.
\begin{lemma}\label{lem:adjmod}
Let $D\in\sa{2}{\K}$ be an arbitrary positive definite matrix. The determinant 
  of the restricted map   
  $\rightabs{\gz{\Lm{D^{1/2}}\circ\Rm{D^{1/2}}}}_{\sa{2}{\K}}$ is 
\begin{equation*}
\det\gz{\rightabs{\gz{\Lm{D^{1/2}}\circ\Rm{D^{1/2}}}}_{\sa{2}{\K}}}\
=\det (D)^{2-\frac{d}{2}},
\end{equation*}
  where $d=\dim_\R \K =1,2$.
\end{lemma}
\begin{proof}
In the real case, 
  $\widetilde{\sa{2}{\R}}=\R \begin{pmatrix}0 & -1 \\1 & 0\end{pmatrix}$. 
One can verify that 
\begin{equation*}
\Lm{D^{1/2}}\circ\Rm{D^{1/2}} 
  =\det(D)^{1/2}\begin{pmatrix}0 & -1 \\1 & 0 \end{pmatrix}
\end{equation*}
  hence
\begin{align*}
\det\gz{\rightabs{\gz{\Lm{D^{1/2}}\circ\Rm{D^{1/2}}}}_{\sa{2}{\R}}} 
&=\frac{\det\gz{\Lm{D^{1/2}}\circ\Rm{D^{1/2}}}}{\det
  \gz{\rightabs{\gz{\Lm{D^{1/2}}\circ\Rm{D^{1/2}}}}_{\widetilde{\sa{2}{\R}}}}}\\ 
&=\frac{\det (D)^2}{\det (D)^{1/2}} = \det (D)^{3/2}.
\end{align*}
In the complex case, we have an isomorphism 
  $\widetilde{\sa{2}{\C}}=\ci\sa{2}{\C}$ and thus
\begin{equation*}
\det\gz{\rightabs{\gz{\Lm{D^{1/2}}\circ\Rm{D^{1/2}}}}_{\sa{2}{\R}}} 
  =\sqrt{\det \gz{\Lm{D^{1/2}}\circ\Rm{D^{1/2}}}} = \det (D)
\end{equation*}
  which completes the proof.
\end{proof}

Note that the Jacobian of the transformation
\begin{equation}
\rightabs{\gz{\Lm{D^{1/2}}\circ\Rm{D^{1/2}}}}_{\sa{2}{\K} }:
  \sa{2}{\K}\to \sa{2}{\K}
\end{equation}
  is $\det (D)^{2d-d^2/2}$ because $D$ is regarded in integral transformations as a
  $4\times 4$ real matrix that acts on $\R^4\cong\C^2$.

\begin{lemma}\label{lem:svratio}
Let $A$ be a $2\times 2$ invertible matrix with singular values 
  $\sigma_1>\sigma_2>0$.
The operator norm and the singular value ratio of $A$, which is defined as 
  $\sigma (A):=\sigma_2/\sigma_1$, can be expressed as follows
\begin{align*}
\norm{A} &= \sqrt{\abs{\det(A)}}
  e^{\frac{1}{2}\cosh^{-1}\gz{\frac{1}{2}\frac{\hsnorm{A}^2}{|\det (A)|}}}\\
\sigma (A)=\frac{\sigma_2}{\sigma_1} &
  =e^{-\cosh^{-1}\gz{\frac{1}{2}\frac{\hsnorm{A}^2}{|\det (A)|}}},
\end{align*}
  where $\hsnorm{\cdot}$ denotes the Hilbert-Schmidt norm.
\end{lemma}
\begin{proof}
By definition, singular values of $A$ are the eigenvalues of $\sqrt{A^\ast A}$
  that are 
\begin{align*}
\sigma_{1,2}&=\sqrt{|\det (A)|}
  \gz{\frac{\hsnorm{A}^2}{2|\det (A)|}\pm
  \sqrt{\gz{\frac{\hsnorm{A}^2}{2|\det (A)|}}^2-1}}^{1/2} \\
&=\sqrt{|\det (A)|}e^{\pm\frac{1}{2}
  \cosh^{-1}\gz{\frac{1}{2}\frac{\hsnorm{A}^2}{|\det (A)|}}}
\end{align*}
  which completes the proof.
\end{proof}

The standard unit ball in the normed vector space of $2\times 2$ matrices is 
  denoted by $\BK$ and the notation $\partial\BK$ stands for the surface of 
  the unit ball. 
We set the notation $\1_A$ for the indicator function of the set
  $A\subseteq\K^{2\times 2}$.

\begin{definition}
The functions  $\chi_d,\eta_d : [0,\infty)\to [0,\infty)$
  are defined by the following formulas
\begin{align}
\chi_d(\varepsilon)&=\int\limits_{\BK}
  \1_{\norm{V_{\varepsilon}^{-1} X V_{\varepsilon}}<1}\dint\la{4d}(X),\\
\eta_d(\varepsilon)&=\int\limits_{\BK}
  \det(I-XX^\ast)^{-\frac{3d}{4}-\frac{1}{2}}
  \1_{\norm{V_{\varepsilon}^{-1} X V_{\varepsilon}}<1}\dint\la{4d}(X),
\end{align}
  where $V_{\varepsilon}=\begin{pmatrix}1&0\\0&\varepsilon\end{pmatrix}$ 
  and $d=\dim_\R (\K)$. 
\end{definition}
Clearly, these functions are reciprocal symmetric i.e.
  $\chi_d (1/\varepsilon)=\chi_d(\varepsilon)$ and
  $\eta_d (1/\varepsilon)=\eta_d (\varepsilon)$ holds for $\varepsilon>0$.
The normalized $\chi_d$-function 
  $\tilde{\chi}_d(\varepsilon)=\chi_d(\varepsilon)/\chi_d(1)$
  measures the probability that a uniformly distributed matrix in   
  $\BK$ is mapped in $\BK$ by the similarity transformation 
  $V_{\varepsilon}^{-1} (.) V_{\varepsilon}$. 
The normalized $\tilde{\eta}_d(\varepsilon)=\eta_d (\varepsilon)/\eta_d (1)$ 
  function would have a similar probabilistic meaning, but we will see that 
  $\eta_d (1)=\infty$ therefore we should find an other way to calculate 
  $\tilde{\eta}_d(\varepsilon)$. 

\begin{lemma}\label{lem:chi}
The function $\tilde{\chi}_1 (\varepsilon ):\sz{0,1}\to\sz{0,1}$
  can be expressed as follows
\begin{align}
\begin{split}
\tilde{\chi}_1 (\varepsilon )&=1-\frac{4}{\pi^2}\int\limits_\varepsilon^1 
\gz{s+\frac{1}{s}-\frac{1}{2}\gz{s-\frac{1}{s}}^2\log \gz{\frac{1+s}{1-s}}}
  \frac{1}{s}\dint s \\
&=\frac{4}{\pi^2}\int\limits_0^\varepsilon
\gz{s+\frac{1}{s}-\frac{1}{2}\gz{s-\frac{1}{s}}^2\log \gz{\frac{1+s}{1-s}}}
  \frac{1}{s}\dint s. 
\end{split}
\end{align}
\end{lemma}
\begin{proof} 
The proof, which is elementary but somewhat lengthy, can be found in 
  Appendix \ref{App:AppendixA}.
\end{proof}

The function $\tilde{\chi}_1 (\varepsilon )$ can be written in a closed form
  using polylogarithmic functions but this is unnecessary for our purposes.
It is somewhat interesting, that the identity function approximates well
  $\tilde{\chi}_1 (\varepsilon )$ (See Fig. \ref{fig:1}).
\begin{figure}[!ht]
\centering
\includegraphics[width=0.75\linewidth]{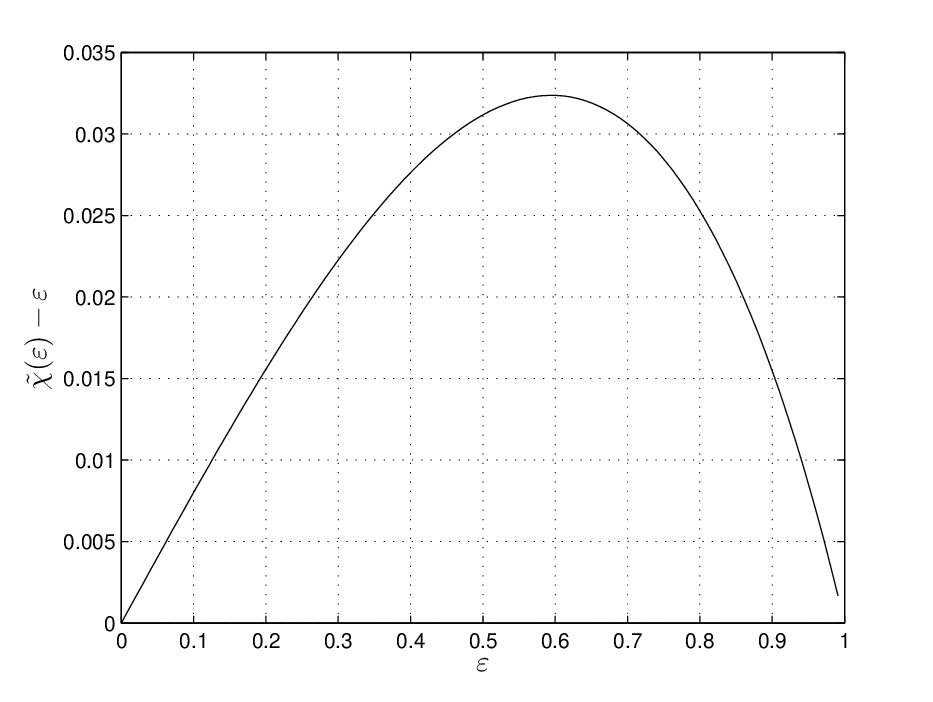}
\label{fig:1}
\caption{Graph of the function $\varepsilon\mapsto\tilde{\chi}_1 (\varepsilon)-\varepsilon$.}
\end{figure}

Recall that Pauli matrices 
$\sigma_1=\begin{pmatrix} 0& 1 \\ 1& 0\end{pmatrix}$, 
$\sigma_2=\begin{pmatrix} 0& -\ci \\ \ci& 0\end{pmatrix}$ and
$\sigma_3=\begin{pmatrix}1 & 0 \\0 & -1\end{pmatrix}$ with 
  $I=\begin{pmatrix}1 & 0 \\0 & 1\end{pmatrix}$
  form an orthogonal basis of the space of $2\times 2$ self-adjoint matrices. 
We parametrize the spaces $\sa{2}{\R}$ and $\sa{2}{\C}$ in the following way
\begin{align}
R(\theta,x,y) &= \frac{x+y}{2}I+\frac{x-y}{2}(\cos(\theta)    
  \sigma_1+\sin(\theta) \sigma_3), \label{eq:sa1}\\
& 0<\theta<2\pi,\, x,y\in\R \nonumber \\ 
R(\theta,\phi,x,y)&=\frac{x+y}{2}I+\frac{x-y}{2}(\cos(\theta)\sin(\phi) 
  \sigma_1+\sin(\theta)\sin(\phi)\sigma_2+\cos(\phi)\sigma_3), \label{eq:sa2}\\
& 0<\theta<2\pi,\, 0<\phi<\pi,\, x,y\in\R. \nonumber
\end{align}
This parametrization is very convenient because eigenvalues of
  $R(\theta,x,y)$ and $R(\theta,\phi,x,y)$ can be directly read out from the 
  parametrization.

We introduce the notation $O(\phi)$ for the standard $2\times 2$ rotation 
  matrix that rotates points counter-clockwise through an angle $\phi$ about 
  the origin.
Let us denote by $\Lambda (x,y)$ the $2\times 2$ diagonal matrix that contains
  $x,y$ in its diagonal. 

Let us introduce the parametrization of $U(2)$ \eqref{eq:U2} that can be found 
  in Mirman's book  (See p. 284--285 in $\cite{Mirman}$).
\begin{align}\label{eq:U2}
\begin{split}
&U(\Theta,\Phi,\omega,\tau) =
e^{\ci\Theta} \times
\begin{pmatrix}
e^{\frac{\ci (\omega+\tau)}{2}}\cos \frac{\Phi}{2} 
  & \ci e^{\frac{\ci (\omega-\tau)}{2}}\sin \frac{\Phi}{2}\\
\ci e^{-\frac{\ci (\omega-\tau)}{2}}\sin \frac{\Phi}{2} 
  & e^{-\frac{\ci (\omega+\tau)}{2}}\cos \frac{\Phi}{2}
\end{pmatrix} \\
&0<\Phi<\pi,\, 0<\Theta<2\pi,\, 0<\omega,\,\tau<4\pi
\end{split}
\end{align}
Polar decomposition will be utilized to parametrize the space of $2\times 2$ 
  complex matrices. 
The space of $2\times 2$ real and complex density matrices will be 
  parametrized by the canonical Bloch sphere parametrization as follows.
\begin{align}
D(\theta,r) &= \frac{1}{2}(I+r(\cos(\theta)\sigma_1+\sin (\theta)\sigma_3)), 
  \label{eq:DR}\\
            & 0<\theta<2\pi,\, 0<r<1 \nonumber \\
D(\theta,\phi,r) &= \frac{1}{2}(I+r(\cos(\theta) \sin (\phi) \sigma_1
  +\sin(\theta)\sin (\phi)\sigma_2+\cos(\phi) \sigma_3)), \label{eq:DC}\\
            & 0<\theta<2\pi,\, 0<\phi<\pi,\, 0<r<1\nonumber
\end{align}

In Table \ref{tab:detg}, we collected the parameterizations of
  manifolds $\R^{2\times 2}$, $\C^{2\times 2}$, $\sa{2}{\R}$, $\sa{2}{\C}$,
  $\DN{2,\R}$, $\DN{2,\C}$ and volume forms corresponding to the considered 
  parametrization. 
These formulas will be applied in the sequel without mentioning.

\begingroup
\renewcommand{\arraystretch}{1.5}
\begin{table}[!ht]
\centering
\begin{tabular}{|c|m{6.5cm}|c|}
\hline
Manifold & Parametrization & Volume form \\ 
\hline
$\R^{2\times 2}$ & $O(\phi)\Lambda (x,y) O(\theta)$, 
  \mbox{$0<\phi,\theta<2\pi$, $0<x,y$}& $\frac{|x^2-y^2|}{2}$ \\ 
\hline
$\C^{2\times 2}$ & $R(\theta,\phi,x,y) U(\Theta,\Phi,\omega,\tau)$, 
  where $0<x,y$ (See Equations \eqref{eq:sa2} and \eqref{eq:U2}.)
  & $\frac{x y (x^2-y^2)^2}{64}\sin \phi \sin\Phi$\\ 
\hline
$\sa{2}{\R}$     & $R(\theta,x,y)$  (See Equation \eqref{eq:sa1})
  & $\frac{|x-y|}{\sqrt{2}}$\\ 
\hline
$\sa{2}{\C}$     & $R(\theta,\phi,x,y)$ (See Equation \eqref{eq:sa2}.)
  & $\frac{(x-y)^2}{2}\sin \phi$ \\ 
\hline
$\DN{2,\R}$      & $D(\theta,r)$ (See Equation \eqref{eq:DR}.) 
  & $\frac{r}{2}$\\ 
\hline
$\DN{2,\C}$      & $R(\theta,\phi,x,y)$ (See Equation \eqref{eq:DC}.)
  & $\frac{r^2\sin\phi}{2\sqrt{2}}$\\ 
\hline
\end{tabular}
\label{tab:detg}
\caption{Parametrization of manifolds $\R^{2\times 2}$, $\C^{2\times 2}$,
  $\sa{2}{\R}$, $\sa{2}{\C}$ and the corresponding volume forms.}
\end{table}
\endgroup

In Table \ref{tab:normc}, we summarize the normalization constants 
  corresponding to $\chi_d$ and $\eta_d$ $d=1,2$.
\begingroup
\renewcommand{\arraystretch}{1.5}
\begin{table}[!ht]
\centering
\begin{tabular}{|c|c|c|}
\hline 
& $d=1$       & $d=2$ \\ \hline
$\chi_d (1) =$ & $\frac{2}{3}\pi^2$  & $\frac{\pi^4}{6}$ \\ \hline
$\eta_d (1) =$ & $\infty$    & $\infty$ \\ \hline
\end{tabular}
\label{tab:normc}
\caption{Normalization constants corresponding to $\chi_d$ and $\eta_d$
  $d=1,2$.}
\end{table}
\endgroup

As an example, we calculate here $\chi_2 (1)$.
By definition, we can write
\begin{equation*}
\chi_2 (1)=\int\limits_{\BC}
  \1_{\Vert V_{1}^{-1} X \underbrace{V_{1}}_{=\id_{\C^2}}\Vert <1}
  \dint\la{4d}(X)
=\int\limits_{\BC}1\dint\la{4d}(X).
\end{equation*}
Now we apply the parametrization and volume form presented in Table 
  \ref{tab:detg} and we obtain
\begin{align*}
\chi_2 (1)&= \int\limits_{\BC}1\dint\la{4d}(X) \\
&=4^3\pi^4\int\limits_0^1\int\limits_0^1\int\limits_0^\pi 
  \int\limits_0^\pi\frac{x y (x^2-y^2)^2}{64}\sin \phi \sin\Phi
  \dint\phi\dint\Phi\dint y\dint x \\
&=4\pi^4\int\limits_0^1\int\limits_0^1 xy(x^2-y^2)^2 \dint y\dint x 
  =\frac{\pi^4}{6}.
\end{align*}

To make the explanation precise, we define $\tilde{\eta}_d(\varepsilon)$ as
\begin{equation*}
\tilde{\eta}_d(\varepsilon)=\lim\limits_{\delta\to 1-0}
  \frac{\int\limits_{\BK}
  \det (I-XX^\ast)^{-\gz{\frac{3d}{4}-\frac{1}{2}}\delta}
  \1_{\norm{V_{\varepsilon}^{-1} X V_{\varepsilon}}<1}
  \dint\la{4d}(X)}{\int\limits_{\BK}
  \det (I-XX^\ast)^{-\gz{\frac{3d}{4}-\frac{1}{2}}\delta}\dint\la{4d}(X)}
\end{equation*}
  which limit exists because the measures
\begin{align*}
&\frac{\gz{(1-x^2)(1-y^2)}^{-\frac{5\delta}{4}}\abs{x^2-y^2}}
  {\int\limits_0^1\int\limits_0^1\gz{(1-t^2)(1-s^2)}^{-\frac{5\delta}{4}}
  \abs{t^2-s^2}\dint\la{2}(t,s)} \dint\la{2}(x,y) \\
&\frac{\gz{(1-x^2)(1-y^2)}^{-2\delta}xy(x^2-y^2)^2}
  {\int\limits_0^1\int\limits_0^1\gz{(1-t^2)(1-s^2)}^{-2\delta}
  st(s^2-t^2)^2 \dint\la{2}(t,s)}  \dint\la{2}(x,y)
\end{align*}
  converge in weak-$\ast$ topology to a measure concentrated on 
  $\kz{(x,y)\in\sz{0,1}|\ x=1 \,\vee\, y=1}$ as $\delta\to 1-0$. 
By the unitary symmetry, we can conclude that the measure
\begin{align*}
\frac{\det (I-XX^\ast)^{-\gz{\frac{3d}{4}-\frac{1}{2}}\delta}}
  {\int\limits_{\BK}
  \det (I-XX^\ast)^{-\gz{\frac{3d}{4}-\frac{1}{2}}\delta}
  \dint\la{4d}(X)}
  \dint\la{4d}(X)
\end{align*}
  converges in weak-$\ast$ topology to the uniform distribution
  on $\partial\BK$ as $\delta\to 1-0$.
The next lemma states that with this definition we get back $\tilde{\chi}_1$.
We conjecture that this identity also holds true for $\tilde{\chi}_2$
  and $\tilde{\eta}_2$.
\begin{lemma}\label{lem:eta}
The functions $\tilde{\chi}_1$ and $\tilde{\eta}_1$ are equals to each other.
\begin{equation*}
\tilde{\chi}_1(\varepsilon)=\tilde{\eta}_1(\varepsilon) 
  \quad\varepsilon\in \sz{0,1}
\end{equation*}
\end{lemma}
\begin{proof}
The proof of this theorem is provided in Appendix \ref{App:AppendixB}.
\end{proof}

\begin{definition}
The \emph{polylogarithmic function} is defined by the infinite sum
\begin{equation*}
\Li_s(z)=\sum\limits_{k=1}^\infty \frac{z^k}{k^s}.
\end{equation*}
  for arbitrary complex $s$ and for all complex arguments $z$ with $|z|<1$.
\end{definition}

\section{The main result}

We parametrize the space of $4\times 4$ density matrices ($\DN{4,\K}$) in the 
  following way 
\begin{equation*}
\rho(D_1,D_2,C)=\begin{pmatrix}D_1    & C \\ C^\ast & D_2 \end{pmatrix},
\end{equation*}
  where $D_1,D_2>0$, $D_1+D_2\in\DN{2,\K}$ and $C\in\K^{2\times 2}$.
Note that, with this parametrization 
\begin{equation}
\Tr_2 (\rho(D_1,D_2,C)) = D_1 + D_2. 
\end{equation}
For a given state $D\in\DN{2,\K}$ we define
\begin{equation}
\DN{4,\K}(D) = \kz{\rightabs{\rho\in\DN{4,\K}}\Tr_2 (\rho) = D},
\end{equation}
  that is the set of those states, which partial trace with respect to the
  system $2$, respectively yield the matrix $D\in\DN{2,\K}$.

Let us introduce the involution
\begin{equation}
\rho(D_1,D_2,C)\mapsto T(\rho(D_1,D_2,C))=\rho(D_1,D_2,C^\ast)
\end{equation} 
  that is just the composition of partial transpose and element-wise 
  conjugation which is a positive map.
Consequently, the aforementioned Peres--Horodecki positive partial transpose
  criterion can be reformulated as
\begin{equation}
\DN{4,\K}^s = T\gz{\DN{4,\K}}\cap\DN{4,\K}.
\end{equation}
Now we are in the position to state one of our main results.

\begin{theorem}\label{thm:1}
Let $D\in\DN{2,\K}$ be a fixed density matrix.
The Hilbert-Schmidt measure of  $\DN{4,\K}^s(D)$ is
\begin{align}\label{eq:cndVol}
\begin{split}
\Vol\gz{\DN{4,\K}^s(D)}&= 
\frac{\det (D)^{4d-\frac{d^2}{2}}}{2^{6d}} \\
&\times \int\limits_{\E{2,\K}}\det (I-Y^2)^d\times
  \chi_d\circ\sigma\gz{\sqrt{\frac{I-Y}{I+Y}}}\dint\la{d+2}(Y)
\end{split}
\end{align}
  and the volume of the space $\DN{4,\K}^s$ can be expressed as
\begin{equation}
\Vol\gz{\DN{4,\K}^s}=\int\limits_{\DN{2,\K}}
  \Vol\gz{\DN{4,\K}^s(D)} \dint\la{d+1}(D),
\end{equation}
  where $d=\dim_\R (\K)=1,2$.
\end{theorem}
\begin{proof}
For fixed $D_1,D_2\in\sa{2}{\K}$, we set
\begin{equation*}
\mathcal{C}(D_1,D_2)=
\kz{\rightabs{C\in\K^{2\times 2}}\rho(D_1,D_2,C)>0,\ \rho(D_1,D_2,C^\ast)>0}.
\end{equation*}
By Fubini's theorem, we have
\begin{align}\label{eq:int}
\Vol\gz{\DN{4,\K}^s} 
&= \la{6d+3}\gz{T\gz{\DN{4,\K}}\cap\DN{4,\K}} \nonumber\\
&= \int\limits_{\begin{array}{c} D_1,D_2>0 \\ \Tr (D_1+D_2) = 1 \end{array}} 
  \int\limits_{C\in\mathcal{C}(D_1,D_2)} 1 \dint\la{4d}(C)
  \dint\la{2d+3}(D_1,D_2).
\end{align}

If $D_1,D_2>0$ and $\Tr (D_1+D_2) = 1$ fixed, then a matrix  
  $C\in\K^{2\times 2}$ belongs to the set $\mathcal{C}(D_1,D_2)$ if and only if
  $\begin{pmatrix}D_1    & C \\ C^\ast & D_2 \end{pmatrix}>0$ and 
  $\begin{pmatrix}D_1    & C^\ast \\ C & D_2 \end{pmatrix}>0$ holds. 
This condition can be reformulated by Lemma \ref{lem:schur}. as 
\begin{align*}
I &> \gz{D_1^{-1/2} C D_2^{-1/2}}^\ast D_1^{-1/2} C D_2^{-1/2} \Leftrightarrow 
  \norm{D_1^{-1/2} C D_2^{-1/2}}<1\\
I &> \gz{D_2^{-1/2} C D_1^{-1/2}}^\ast D_2^{-1/2} C D_1^{-1/2} \Leftrightarrow 
  \norm{D_2^{-1/2} C D_1^{-1/2}}<1,
\end{align*}
  where $\norm{\cdot}$ denotes the usual operator norm.

To compute the inner integral, we substitute 
\begin{equation*}
X=D_1^{-1/2} C D_2^{-1/2}=\gz{\Lm{D_1^{-1/2}} \circ \Rm{D_2^{-1/2}}} (C).
\end{equation*}
The Jacobian of this transform is 
  $\det\gz{\Lm{D_1^{-1/2}}\circ\Rm{D_2^{-1/2}}}^{-1}=\det (D_1)^d\det (D_2)^d$ 
  and the inner integral of \eqref{eq:int} can be written as
\begin{equation*}
\int\limits_{C\in\mathcal{C}(D_1,D_2)}1\dint\la{4d}(C) 
  =\det (D_1 D_2)^d\int\limits_{\BK}
  \1_{\norm{(V^\ast)^{-1} X V}<1}\dint\la{4d}(X),
\end{equation*}
  where $V=D_2^{1/2}D_1^{-1/2}$.
Observe that the last term depends only on the  singular value ratio of $V$, 
  because taking the singular value decomposition of $V$: $V=U_1\Sigma U_2$,
  we have
\begin{equation*}
\norm{(V^\ast)^{-1} X V}=\norm{U_1\Sigma^{-1}U_2 X U_1\Sigma U_2}
  =\norm{\Sigma^{-1}U_2 X U_1\Sigma}
\end{equation*} 
  and the transformation $X\mapsto U_2 X U_1$ is isometric with respect to the 
  Hilbert-Schmidt norm.
It means that
\begin{equation*}
\chi_d (\sigma (V))=\int\limits_{\BK}
  \1_{\norm{(V^\ast)^{-1} X V}<1} \dint\la{4d}(X)
\end{equation*}
holds. 
By Lemma \ref{lem:svratio}, the singular value ratio of $V$ is
\begin{equation*}
\sigma(V)=e^{-\cosh^{-1}\gz{\frac{\hsnorm{V}^2}{2|\det (V)|}}}
  = e^{-\cosh^{-1}\gz{\frac{1}{2}\sqrt{\frac{\det (D_1)}{\det (D_2)}}
  \Tr \gz{D_2 D_1^{-1}}}}
\end{equation*}
  hence for the volume of separable states \eqref{eq:int} we obtain
\begin{equation*}
\Vol\gz{\DN{4,\K}^s}
  =\int\limits_{\begin{array}{c}D_1,D_2>0 \\ \Tr (D_1+D_2) = 1\end{array}}  
  \det (D_1 D_2)^d f(D_2 D_1^{-1})\dint\la{2d+3}(D_1,D_2), 
\end{equation*}
  where 
\begin{equation*}
f(D_2 D_1^{-1})=\chi_d\circ\exp 
  \gz{{-\cosh^{-1}\gz{\frac{1}{2}\sqrt{\frac{\det (D_1)}{\det (D_2)}}
  \Tr \gz{D_2 D_1^{-1}}}}}.
\end{equation*}
We introduce the parametrization
\begin{align}\label{eq:param}
\begin{split}
D_1 = \frac{1}{2}(D + A) \\
D_2 = \frac{1}{2}(D - A),
\end{split}
\end{align}
  where $D$ takes values in $\DN{2,\K}$ and $A$ runs on self-adjoint 
  $2\times 2$ matrices that satisfy the condition $-D<A<D$.

By the invariance of trace under cyclic permutations, the previous integral can 
  be written as
\begin{equation*}
\Vol\gz{\DN{4,\K}^s}=\int\limits_{\DN{2,\K}}\Vol\gz{\DN{4,\K}^s(D)}
  \dint\la{d+1}(D),
\end{equation*}
  where
\begin{align*}
&\Vol\gz{\DN{4,\K}^s(D)} =\frac{\det (D)^{2d}}{2^{6d}}\times \\
&\times
\int\limits_{\begin{array}{c}A\in\sa{2}{\K} \\-D<A<D\end{array}}
  \det (I-(D^{-1/2}AD^{-1/2})^2)^d 
  f\gz{\frac{I-D^{-1/2}AD^{-1/2}}{I+D^{-1/2}AD^{-1/2}}} \dint\la{d+2}(A).
\end{align*}
We substitute $Y=D^{-1/2}AD^{-1/2}=\gz{\Lm{D^{-1/2}}\circ\Rm{D^{-1/2}}}(A)$.
According to the remark after Lemma \ref{lem:adjmod}, the Jacobian of this 
  transformation is 
\begin{equation*}
\det\gz{\Lm{D^{-1/2}}\circ\Rm{D^{-1/2}}}^{-1} = \det (D)^{2d-d^2/2}.
\end{equation*}
Observe that 
  $f\gz{\frac{I-Y}{I+Y}}= \chi_d\circ\sigma\gz{\sqrt{\frac{I-Y}{I+Y}}}$ and thus
\begin{equation*}
\Vol\gz{\DN{4,\K}^s(D)} = 
  \frac{\det (D)^{4d-\frac{d^2}{2}}}{2^{6d}} \int\limits_{\E{2,\K}}
  \det (I-Y^2)^d\times (\chi_d\circ\sigma)\gz{\sqrt{\frac{I-Y}{I+Y}}}
  \dint\la{d+2}(Y)
\end{equation*}
  which completes the proof.
\end{proof}

The next Corollary proves Milz and Strunz's conjecture on the behavior of the 
  conditioned volume over reduced states (See equation (23) in \cite{Milz2015}).
\begin{corollary}
In the complex case, the conditioned volume can be expressed as 
\begin{equation*}
\Vol\gz{\DN{4,\K}^s(D)}= K_1\times\det (D)^6 = K_2\times (1-r^2)^6,
\end{equation*}
  where $K_1,K_2$ are constants and $r$ is the radius of $D$ in the 
  Bloch sphere.
\end{corollary}
\begin{proof}
We set $d=2$ in \eqref{eq:cndVol} and obtain the first equality.
According to the parametrization of $\DN{2,\C}$ \eqref{eq:DC}, 
  $\det (D)=\frac{1}{4}(1-r^2)$ which proves the second equality.
\end{proof}

\begin{corollary}\label{cor:HS}
If $D\in\DN{2,\K}$ is a fixed density matrix, then the probability to find a 
  separable state in $\DN{4,\K}(D)$ can be written as
\begin{equation}\label{eq:psep}
\mathcal{P}_{sep}(\K)=\int\limits_{\E{2,\K}} 
  \tilde{\chi}_d\circ\sigma\gz{\sqrt{\frac{I-Y}{I+Y}}}\dint\mu_{d+2}(Y), 
\end{equation}
  where
\begin{equation*}
\dint\mu_{d+2}(Y)= 
  \frac{\det (I-Y^2)^d}{\int\limits_{\E{2,\K}}\det (I-Z^2)^d\dint\la{d+2}(Z)}
  \dint\la{d+2}(Y).
\end{equation*}
It is apparent that this probability is not depend on $D$ that proves the 
  conjecture of Milz and Strunz \cite{Milz2015}. 
\end{corollary}
\begin{proof}
In a similar way, we can calculate the volume of the whole space
\begin{equation*}
\Vol\gz{\DN{4,\K}(D)}=\chi_d(1)\frac{\det (D)^{4d-\frac{d^2}{2}}}{2^{6d}} 
  \int\limits_{\E{2,\K}}\det (I-Y^2)^d\dint\la{d+2}(Y) 
\end{equation*}
  and thus we have
\begin{align*}
\frac{\Vol\gz{\DN{4,\K}^s(D)}}{\Vol\gz{\DN{4,\K}(D)}}
=\frac{\int\limits_{\E{2,\K}}\det (I-Y^2)^d\,
  \chi_d\circ\sigma\gz{\sqrt{\frac{I-Y}{I+Y}}} \dint\la{d+2}(Y)}
  {\chi_d (1)\int\limits_{\E{2,\K}}\det (I-Y^2)^d\dint\la{d+2}(Y)}
\end{align*}
  which completes the proof.
\end{proof}

Using the fact that $\mu_{d+2}$ and $\sigma\gz{\sqrt{\frac{I-Y}{I+Y}}}$ are 
  invariant under orthogonal (unitary) transformation, we can simplify
  \eqref{eq:psep} and we obtain the following theorem.

\begin{theorem}
The separability probability in the rebit-rebit system with respect to the
  Hilbert--Schmidt measure is
\begin{equation}
\mathcal{P}_{\text{sep}}(\R)=\frac{29}{64}.
\end{equation}
\end{theorem}
\begin{proof}
According to Equation \eqref{eq:psep} and the above mentioned unitary 
  invariance, we can write
\begin{equation*}
\mathcal{P}_{sep}(\R) =
\frac{\int\limits_{-1}^1\int\limits_{-1}^x \tilde{\chi}_1 
  \gz{\sqrt{\frac{1-x}{1+x}}\times \sqrt{\frac{1+y}{1-y}}}
      (1-x^2)(1-y^2)(x-y)\dint y\dint x}
  {\int\limits_{-1}^1\int\limits_{-1}^x (1-x^2)(1-y^2)(x-y)\dint y\dint x},
\end{equation*}
  where the denominator is equal to $\frac{16}{35}$. In the numerator, we   
  substitute $u=\frac{1-x}{1+x}$, $v=\frac{1-y}{1+y}$. 
The Jacobian of this transformation is $\frac{4}{(1+u)^2(1+v)^2}$. 
Not that, the map $z\mapsto \frac{1-z}{1+z}$ is a monotone decreasing 
  involution that maps $(-1,1)$ onto $(0,\infty)$.
After substitution, the numerator gains the following form
\begin{equation*}
\int\limits_0^\infty\int\limits_0^v
  \tilde{\chi}_1 \gz{\sqrt{\frac{u}{v}}}\frac{128uv(v-u)}{(1+u)^5(1+v)^5}
  \dint u\dint v.
\end{equation*}
We substitute again. Let $u=ts$ and $v=\frac{s}{t}$. 
The Jacobian of this substitution is $\frac{2s}{t}$ and the domain of 
  integration is $0<s<\infty$, $0<t<1$. 
So, we have
\begin{align*}
\int\limits_0^\infty\int\limits_0^v
  \tilde{\chi}_1 \gz{\sqrt{\frac{u}{v}}} \frac{128uv(v-u)}{(1+u)^5(1+v)^5}
  \dint u\dint v 
=\int\limits_0^\infty\int\limits_0^1
\tilde{\chi}_1 (t)\frac{256s^4t^3(1-t^2)}{(s+t)^5 (1+st)^5}\dint t\dint s.
\end{align*}
Now we integrate by parts in the inner integral and we get
\begin{equation*}
\int\limits_0^1 \tilde{\chi}_1 (t)
  \frac{256s^4t^3(1-t^2)}{(s+t)^5 (1+st)^5}\dint t 
=\frac{64 s^3}{(s+1)^8}-\int\limits_0^1
  \frac{64s^3t^4 \gz{\tilde{\chi}_1}' (t)}{(s+t)^4(1+st)^4}\dint t.
\end{equation*}
With this, the numerator can be written as
\begin{align*}
&\int\limits_0^\infty\frac{64 s^3}{(s+1)^8}\dint s
  -\int\limits_0^1\int\limits_0^\infty
  \frac{64s^3t^4 \gz{\tilde{\chi}_1}' (t)}{(s+t)^4(1+st)^4}
  \dint s\dint t\\
&=\frac{16}{35}- \frac{64}{3}\int\limits_0^1
  \frac{11(1-t^6)+27t^2(1-t^2)+6(1+t^2)(1+8t^2+t^4)\log (t)}{(t^2-1)^7}
  \gz{\tilde{\chi}_1}' (t)\dint t,
\end{align*}
  where we have interchanged the order of integration in the last term.

One can check that
\begin{align*}
\frac{64}{3}
  \int&\frac{11(1-t^6)+27t^2(1-t^2)+6(1+t^2)(1+8t^2+t^4)\log (t)}{(t^2-1)^7} 
  \gz{\tilde{\chi}_1}' (t)\dint t = \\
=&-\frac{1}{9\pi ^2\gz{t^2-1}^6}
  \leftsz{9\gz{t^2-1}^6 \Li_2(1-t)+9\gz{t^2-1}^6 \Li_2(-t)}+ \\
&+96\gz{t^2+1}\gz{t^4+28 t^2+1}\gz{t^2-1}^3\tanh ^{-1}(t)+ \\
&+9\gz{t^8-132 t^6-378 t^4-132 t^2+1}\gz{t^2-1}^2 \log (t)\log (t+1)+ \\
&+2t\gz{-57 t^{10}-1211 t^8+78 t^6-78 t^4+1211 t^2}+ \\
&+6t\leftgz{\leftgz{-3 t^{10}+401 t^8+882 t^6+882 t^4+401 t^2+192 +}} \\
&+\rightsz{\rightgz{\rightgz{\gz{t^9+t^7-4t^5+t^3+t} \log (1-t)-3}\log (t)+57}}
 +\mbox{const},
\end{align*}
  where we applied Lemma \ref{lem:chi}. 
Using this, one can conclude that
\begin{equation*}
\frac{64}{3}\int\limits_0^1
  \frac{11(1-t^6)+27t^2(1-t^2)+6(1+t^2)(1+8t^2+t^4)\log (t)}{(t^2-1)^7}
  \gz{\tilde{\chi}_1}' (t)\dint t = \frac{1}{4}
\end{equation*}
  and thus we have
\begin{equation*}
\mathcal{P}_{sep}(\R) = \frac{\frac{16}{35}-\frac{1}{4}}{\frac{16}{35}}
  =\frac{29}{64}
\end{equation*}
which completes the proof.
\end{proof}

\section{Generalization to $\gz{\DN{4,\K},g_{\sqrt{x}}}$}

The operator monotone function $f:\R^+\to R$ is said to be symmetric and 
  normalized if $f(x)=xf(x^{-1})$ holds for every positive argument $x$ and 
  $f(1)=1$.
The set of symmetric and normalized operator monotone functions plays an 
  important role in quantum information geometry \cite{PetzQinf,Dittmann1}. 
Petz's classification theorem states that there exists a bijective
  correspondence between the set of symmetric and normalized operator monotone 
  functions and the family of monotone metrics \cite{PetzSud96}.
The metric associated to the operator monotone function $f$ is given by
\begin{equation}
g_f(D)(X,Y)=\Tr\gz{X
  \gz{\Rm{D}^{\frac{1}{2}}f\gz{\Lm{D}\Rm{D}^{-1}}\Rm{D}^{\frac{1}{2}}}^{-1}(Y)}
\end{equation}
  for all $n\in\N^+$, $D\in\DN{n,\K}$ and $X,Y\in T_D\DN{n,\K}$. 
The space of $n\times n$ density matrices endowed with the accompanying 
  monotone metric of the operator monotone function $f$ is denoted by 
  $\gz{\DN{n,\K},g_f}$.

In this point, we generalize our results to the space
  $\gz{\DN{4,\K},g_{\sqrt{x}}}$.
According to theorem 6 in \cite{AndaiVol}, the volume form of 
  $\gz{\DN{n,\K},g_f}$ can be expressed as
\begin{equation}
\sqrt{\det (g_f(D))}=\frac{1}{\sqrt{\det (D)}}
  \gz{2^{\binom{n}{2}}\prod\limits_{1\le i<j\le n}c_f (\mu_i,\mu_j)}^{d/2},
\end{equation}
  where $d=\dim_\R \K$, $\mu_i$-s are the eigenvalues of $D$ and   
  $c_f(x,y)=\frac{1}{y f(x/y)}$ is the \v{C}enzov--Morozova function 
  associated to $f$. 
For $n=4$ and $f(x)=\sqrt{x}$, we have
\begin{equation}
\sqrt{\det (g_f(D))}
  =\frac{2^{3d}}{\det (D)^{\frac{3d}{4}+\frac{1}{2}}}\quad d=1,2.
\end{equation}

A slight modification of the previous proofs gives the following Theorem.

\begin{theorem}
Let $D\in\DN{2,\K}$ be a fixed density matrix.
The volume of the submanifold $\gz{\DN{4,\K}^s(D),g_{\sqrt{x}}}$ 
  can be formally written as
\begin{align*}
\Vol_{\sqrt{x}}\gz{\DN{4,\K}^s(D)} 
  &=4\det (D)^{\frac{5}{2}d-\frac{d^2}{2}-1} \\
&\times \int\limits_{\E{2,\K}}\det (I-Y^2)^\frac{d-2}{4} \eta_d\circ\sigma
  \gz{\sqrt{\frac{I-Y}{I+Y}}} \dint\la{d+2}(Y)
\end{align*}
  and the volume of the space $\gz{\DN{4,\K}^s,g_{\sqrt{x}}}$ can be formally 
  expressed as
\begin{equation*}
\Vol_{\sqrt{x}}\gz{\DN{4,\K}^s}=\int\limits_{\DN{2,\K}}
  \Vol_{\sqrt{x}}\gz{\DN{4,\K}^s(D)}\dint\la{d+1}(D),
\end{equation*}
  where $d=\dim_\R (\K)=1,2$.
\end{theorem}
\begin{proof}
By the factorization 
  $\det(D)=\det (D_1D_2)\det(I-D_1^{-1/2}CD_2^{-1}C^\ast D_1^{-1/2})$, 
  we can write
\begin{align}
&\Vol_{\sqrt{x}}\gz{\DN{4,\K}^s}
=2^{3d}\int\limits_{\begin{array}{c}
  D_1,D_2>0 \\
  \Tr (D_1+D_2) = 1
  \end{array}} \det (D_1D_2)^{-\frac{3d}{4}-\frac{1}{2}}\\
&\times\int\limits_{C\in\mathcal{C}(D_1,D_2)}
  \det (I-D_1^{-1/2}CD_2^{-1}C^\ast D_1^{-1/2})^{-\frac{3d}{4}-\frac{1}{2}} 
  \dint\la{4d}(C) \dint\la{2d+3}(D_1,D_2),
\end{align}
  where $\mathcal{C}(D_1,D_2)$ is the same as in Theorem \ref{thm:1}.
Using the substitution $X=D_1^{-1/2}CD_2^{-1/2}$, the inner integral
  can be written in the following form
\begin{equation*}
\det (D_1 D_2)^d\int\limits_{\BK}
  \det (I-XX^\ast)^{-\frac{3d}{4}-\frac{1}{2}}\1_{\norm{(V^\ast)^{-1} X V}<1}
  \dint\la{4d}(X),
\end{equation*} 
  where $V=D_2^{1/2}D_1^{-1/2}$. 
By a similar argument, the last term depends only on  the singular value 
  ratio of $V$ hence it can be written as $\eta_d \circ \sigma (V)$.
As a result, for the volume we get
\begin{equation*}
\Vol_{\sqrt{x}}\gz{\DN{4,\K}^s}
=\int\limits_{\begin{array}{c}
  D_1,D_2>0 \\
  \Tr (D_1+D_2) = 1
  \end{array}} 
  2^{3d}\det (D_1D_2)^{\frac{d-2}{4}}\eta_d \circ \sigma (D_2^{1/2}D_1^{-1/2})
  \dint\la{2d+3}(D_1,D_2).
\end{equation*}
Using the parametrization \eqref{eq:param} and the substitution  
  $Y=D^{-1/2}AD^{-1/2}$, we obtain
\begin{equation*}
\Vol_{\sqrt{x}}\gz{\DN{4,\K}^s} =\int\limits_{\DN{2,\K}}
  \Vol_{\sqrt{x}}\gz{\DN{4,\K}^s(D)}\dint\la{d+1}(D),
\end{equation*}
  where
\begin{equation*}
\Vol_{\sqrt{x}}\gz{\DN{4,\K}^s(D)}
  =4\det(D)^{\frac{5}{2}d-\frac{d^2}{2}-1}\int\limits_{\E{2,\K}}
  \det (I-Y^2)^\frac{d-2}{4}\eta_d\circ\sigma 
  \gz{\sqrt{\frac{I-Y}{I+Y}}} \dint\la{d+2}(Y).
\end{equation*}
\end{proof}

\begin{corollary}\label{cor:Geo}
For a fixed density matrix $D\in\gz{\DN{2,\K},g_{\sqrt{x}}}$
  the probability to find a separable state in $\DN{4,\K}(D)$ is
\begin{equation}
\mathcal{P}_{sep,\sqrt{x}}(\K)=\int\limits_{\E{2,\K}} 
  \tilde{\eta}_d\circ\sigma\gz{\sqrt{\frac{I-Y}{I+Y}}}\dint\nu_{d+2}(Y), 
\end{equation}
  where
\begin{equation*}
\dint\nu_{d+2}(Y) = \frac{\det (I-Y^2)^\frac{d-2}{4}}{\int\limits_{\E{2,\K}}
  \det (I-Z^2)^\frac{d-2}{4}\dint\la{d+2}(Z)}\dint\la{d+2}(Y).
\end{equation*}
This probability is also independent from $D$ which means that the conjecture 
  of Milz and Strunz holds true for the statistical manifold 
  $\gz{\DN{4,\K},g_{\sqrt{x}}}$. 
\end{corollary}
\begin{proof}
Similarly, one can show that
\begin{equation}
\label{eq:volcurv}
\Vol_{\sqrt{x}}\gz{\DN{4,\K}(D)}=
  4\eta_d (1)\det (D)^{\frac{5}{2}d-\frac{d^2}{2}-1}
  \times \int\limits_{\E{2,\K}}\det (I-Y^2)^\frac{d-2}{4} \dint\la{d+2}(Y)
\end{equation}
  then we take the ratio
\begin{equation*}
\mathcal{P}_{sep,\sqrt{x}}(\K)
=\frac{\Vol_{\sqrt{x}}\gz{\DN{4,\K}^s(D)}}
  {\Vol_{\sqrt{x}}\gz{\DN{4,\K}(D)}}
\end{equation*}
  and we get the desired result.
\end{proof}

Now, we are in the position to calculate the separability probability for
  rebit-rebit systems in this setting.

\begin{theorem}
The separability probability in the statistical manifold   
  $\gz{\DN{4,\R},g_{\sqrt{x}}}$ is
\begin{equation*}
\mathcal{P}_{sep,\sqrt{x}}(\R)=\int\limits_0^1
  \frac{8 \gz{8 \gz{t^4+t^2} E\gz{1-\frac{1}{t^2}}
  -\gz{t^2+3} \gz{3 t^2+1}K\gz{1-\frac{1}{t^2}}}}{\pi \sqrt{t} \gz{t^2-1}^3}
  \tilde{\chi}_1 (t) \dint t\approx 0.26223,
\end{equation*}
  where $K$ is the complete elliptic integral of the first kind
  and $E$ is the elliptic integral of the second kind, that is
\begin{equation*}
K(k)=\int_{0}^{1}\frac{1}{\sqrt{1-t^{2}}\sqrt{1-k^{2}t^{2}}}\dint t
\qquad\mbox{and}\qquad
E(k)=\int_{0}^{1}\frac{\sqrt{1-k^{2}t^{2}}}{\sqrt{1-t^{2}}}\dint t.
\end{equation*}
\end{theorem}
\begin{proof}
Due to the fact that $\tilde{\eta}_1=\tilde{\chi}_1$ 
  (See Lemma \ref{lem:eta} and Appendix \ref{App:AppendixB}.) 
  and by the unitary invariance, we can write
\begin{equation*}
\mathcal{P}_{sep,\sqrt{x}}(\R) =
  \frac{\int\limits_{-1}^1\int\limits_{-1}^x \tilde{\chi}_1 
  \gz{\sqrt{\frac{1-x}{1+x}}\times\sqrt{\frac{1+y}{1-y}}}
  (1-x^2)^{-\frac{1}{4}}(1-y^2)^{-\frac{1}{4}}(x-y)\dint y\dint x}
{\int\limits_{-1}^1\int\limits_{-1}^x 
  (1-x^2)^{-\frac{1}{4}}(1-y^2)^{-\frac{1}{4}}(x-y)\dint y\dint x},
\end{equation*}
  where the denominator is equal to $\frac{2\pi}{3}$. 
To evaluate the numerator, we use the same strategy that we have applied in 
  the Hilbert--Schmidt case. 
After the first substitution, the numerator gains the following form
\begin{equation*}
\int\limits_0^\infty\int\limits_0^v\tilde{\chi}_1 \gz{\sqrt{\frac{u}{v}}}
  \frac{4(v-u)}{(uv)^{\frac{1}{4}}(1+u)^\frac{5}{2}(1+v)^\frac{5}{2}}
  \dint u\dint v.
\end{equation*}
After the second substitution, we have
\begin{align*}
\int\limits_0^\infty\int\limits_0^v\tilde{\chi}_1 \gz{\sqrt{\frac{u}{v}}}
  \frac{4(v-u)}{(uv)^{\frac{1}{4}}(1+u)^\frac{5}{2}(1+v)^\frac{5}{2}}
  \dint u\dint v
=\int\limits_0^\infty\int\limits_0^1
  \frac{8s^{\frac{3}{2}}\sqrt{t}(1-t^2)\tilde{\chi}_1 (t)}
  {(s+t)^\frac{5}{2}(1+ts)^\frac{5}{2}}\dint t\dint s.
\end{align*}
We interchange the order of integration and obtain
\begin{equation*}
\int\limits_0^1
\frac{16 \gz{8 \gz{t^4+t^2} E\gz{1-\frac{1}{t^2}}
  -\gz{t^2+3} \gz{3 t^2+1}K\gz{1-\frac{1}{t^2}}}}{3 \sqrt{t} \gz{t^2-1}^3}
  \tilde{\chi}_1 (t)\dint t\approx 0.549213
\end{equation*}
  that can be evaluate only numerically.
For the separability probability, we have
\begin{equation*}
\mathcal{P}_{sep,\sqrt{x}}(\R)=\int\limits_0^1
  \frac{8 \gz{8 \gz{t^4+t^2} E\gz{1-\frac{1}{t^2}}
  -\gz{t^2+3} \gz{3 t^2+1}K\gz{1-\frac{1}{t^2}}}}{\pi \sqrt{t} \gz{t^2-1}^3}
  \tilde{\chi}_1 (t)\dint t\approx 0.26223
\end{equation*}
  which completes the proof.
\end{proof}

\section{Examples}

To verify the results, first we calculate the volume of $4\times 4$
  density matrices with respect to the standard Lebesgue measure.
As we mentioned in the proof of Corollary \ref{cor:HS}, the volume of
  $\DN{4,\K}$ can be expressed as
\begin{equation}
\Vol\gz{\DN{4,\K}}=\int\limits_{\DN{2,\K}}
  \Vol\gz{\DN{4,\K}(D)}\dint\la{d+1}(D)
\end{equation}
  which can be written in the following product form
\begin{equation}
\Vol\gz{\DN{4,\K}}=\frac{\chi_d(1)}{2^{6d}}\times 
  \int\limits_{\DN{2,\K}}\det (D)^{4d-\frac{d^2}{2}}\dint\la{d+1}(D)  
  \times\int\limits_{\E{2,\K}}\det (I-Y^2)^d\dint\la{d+2}(Y).
\end{equation}
In the real case we have
\begin{align*}
\chi_1(1) &= \frac{2}{3}\pi^2\\
\int\limits_{\DN{2,\R}}\det(D)^{\frac{7}{2}}\dint\la{2}(D)
  &=\frac{\pi}{2^7 3^2}\\
\int\limits_{\E{2,\R}}\det (I-Y^2)\dint\la{3}(Y)&=\frac{2^5\sqrt{2}\pi}{35}.
\end{align*}
In the complex case we have
\begin{align*}
\chi_2(1)&=\frac{\pi^4}{6}\\
\int\limits_{\DN{2,\C}}\det (D)^{6}\dint\la{3}(D) 
  &=\frac{\pi}{2\times 3^2\times 5\times 7\times 11\times 13\times\sqrt{2}}\\
\int\limits_{\E{2,\C}}\det (I-Y^2)^2\dint\la{4}(Y)
  &=\frac{2^{10}\pi}{3^2\times 5^2\times 7}.
\end{align*}
If we put all together, we get
\begin{align*}
\Vol\gz{\DN{4,\R}} 
  &=\frac{\pi^4}{\sqrt{2}\times 2^6\times 3^3\times 35} \\
\Vol\gz{\DN{4,\C}}
  &=\frac{\pi^6}{\sqrt{2}\times 2^{14}\times 3^4\times 5^3\times 7^2
  \times 11\times 13}
\end{align*}
  which is equal to the volume obtained by \.{Z}yczkowski and Sommers
  \cite{ZyczMixedStates} and Andai  (See Theorem 1 and 2 in \cite{AndaiVol}) 
  up to a factor that comes from the difference between the Lebesgue measure 
  and the Hilbert--Schmidt measure. 
Contrary to the $2\times 2$ case (See Corollary 1 in \cite{AndaiVol}),
  the volume of the statistical manifold $\gz{\DN{4,\K},g_{\sqrt{x}}}$
  is infinite in both of the real and complex cases because
  $\eta_d (1)=\infty$ (See Table \ref{tab:normc}.) and the
  volume admits the following factorization
\begin{equation*}
\Vol_{\sqrt{x}}\gz{\DN{4,\K}}=4\eta_d (1) \times
  \int\limits_{\DN{2,\K}}\det(D)^{\frac{5}{2}d-\frac{d^2}{2}-1}\dint\la{d+1}(D) 
  \times\int\limits_{\E{2,\K}}\det (I-Y^2)^\frac{d-2}{4}\dint\la{d+2}(Y).
\end{equation*}

\section{Conclusion}

The structure of the unit ball in operator norm of $2\times 2$ matrices plays a 
  critical role in separability probability of qubit-qubit and rebit-rebit 
  quantum systems.
It is quite surprising that the space of $2\times 2$ real or complex matrices 
  seems simple, but to compute the volume of the set
\begin{equation*}
\kz{\begin{pmatrix}a & b\\ c& d\end{pmatrix} \Bigm\vert\ 
\norm{\begin{pmatrix} a & b\\ c& d\end{pmatrix}}<1,\ \
\norm{\begin{pmatrix} a & \varepsilon b\\ \frac{c}{\varepsilon}& d
\end{pmatrix}}<1 }
\end{equation*}
  for a given parameter $\varepsilon\in\sz{0,1}$, which is the value of 
  the function $\chi_{d}(\varepsilon)$, is a very challenging problem.
The gist of our considerations is that the behavior of the function 
  $\chi_{d}(\varepsilon)$ determines the separability probabilities with respect
  to the Hilbert-Schmidt measure. 
When the volume form generated by the operator monotone function 
  $x\mapsto\sqrt{x}$, a reasonable normalization can be given to define the
  separability probability and in this case the probability is determined by the  
  structure of the surface of the unit ball.

%\bibliography{sepvol3} 
%\bibliographystyle{plain}

%%% Copy from sepvol3.bbl
\def\polhk#1{\setbox0=\hbox{#1}{\ooalign{\hidewidth
  \lower1.5ex\hbox{`}\hidewidth\crcr\unhbox0}}}

%%% End of copy

\newpage
\appendix

\section{Proof of Lemma \ref{lem:chi}} \label{App:AppendixA}

For local usage, we redefine the matrix 
  $\Lambda_\delta = \begin{pmatrix}1 & 0 \\0 & e^{-\delta}\end{pmatrix}$, 
  where $\delta>0$. 
Let us introduce the function 
  $\Delta (\delta)=\Vol(\BR)-\chi_1(e^{-\delta})$
  to which we will refer as a \emph{defect function}. 
In terms of $\Delta$, the statement of Lemma \ref{lem:chi} can be 
  reformulated as follows for every positive $\delta$
\begin{equation*}
\Delta (\delta)=\frac{16}{3}\int\limits_0^\delta
  \cosh t - \sinh^2 t\log \gz{\frac{e^t+1}{e^t-1}}\dint t.
\end{equation*}

First we fix $\delta >0$ and we cover the space of $2\times 2$ real matrices
  with the following atlas 
\begin{equation}
\mathcal{A}=\kz{X_{\pm}(r,t,\rho,\phi),X_{\pm}(r,t,\rho,\phi)\sigma_3},
\end{equation} 
  where
\begin{align}
X_{\pm}(r,t,\rho,\phi) &= r Y_{\pm}(t,\rho,\phi) \nonumber\\
Y_{\pm}(t,\rho,\phi)   &= 
\begin{pmatrix}
\sqrt{\rho}\cos\phi 
  & \pm\frac{\frac{\rho}{2}\sin 2\phi-1}
  {\sqrt{\abs{\frac{\rho}{2}\sin 2\phi-1}}}e^t\\
  \pm\sqrt{\abs{\frac{\rho}{2}\sin 2\phi-1}}e^{-t} 
  & \sqrt{\rho}\sin\phi
\end{pmatrix}, \label{parR}
\end{align}
  $t\in\R$, $r,\rho>0$ and $\phi\in\left[0,2\pi\right[$.
This parametrization is very convenient because the similarity transformation 
  by $\Lambda_\delta$ is just a translation
\begin{equation*}
(r,t,\rho,\phi) 
  \stackrel{\Lambda_\delta^{-1}(.)\Lambda_\delta}{\Longrightarrow} 
(r,t-\delta,\rho,\phi).
\end{equation*}
The metric tensor ($g$) corresponding to this parametrization ($X_\pm $) has 
  $10$ independent components.
\begin{align*}
g_{rr}    &= \rho + 2\cosh(2t) \abs{\frac{\rho}{2}\sin(2\phi)-1}\\
g_{rt}    &= 2r\sinh(2t) \abs{\frac{\rho}{2}\sin(2\phi)-1}\\
g_{r\rho} &= \frac{r}{2}\gz{1+\sin(2\phi)\cosh(2t)
  \sgn\gz{\frac{\rho}{2}\sin(2\phi)-1}}\\
g_{r\phi} &= r\rho\cos(2\phi)\cosh(2t) 
  \sgn\gz{\frac{\rho}{2}\sin(2\phi)-1}\\
g_{tt}    &= 2r^2\cosh(2t) \abs{\frac{\rho}{2}\sin(2\phi)-1}\\
g_{t\rho} &=\frac{r^2}{4}\sin(2\phi)\cosh(2t) 
  \sgn\gz{\frac{\rho}{2}\sin(2\phi)-1}\\
g_{t\phi} &= -r^2\rho\cos(2\phi)\sinh(2t) 
  \sgn\gz{\frac{\rho}{2}\sin(2\phi)-1}\\
g_{\rho\rho} &= \frac{r^2}{4}\gz{\frac{1}{\rho}+\frac{\cosh(2t)\sin^2(2\phi)}
  {2\abs{\frac{\rho}{2}\sin(2\phi)-1}}}\\
g_{\rho\phi} &= \frac{r^2\rho\cosh(2t)\sin(4\phi)}
  {8\abs{\frac{\rho}{2}\sin(2\phi)-1}}\\
g_{\phi\phi} &= r^2\rho\gz{1+\frac{\rho\cosh(2t)\cos^2(2\phi)}
  {\abs{\frac{\rho}{2}\sin(2\phi)-1}}}
\end{align*}
Although the metric tensor has a complicated form, the volume form is quite 
  simple
\begin{equation*}
\sqrt{\det (g(r,t,\rho,\phi))}=r^3.
\end{equation*}
We can write 
\begin{align*}
\chi_1 (e^{-\delta})&=\la{4}\gz{\BR\cap \Lambda_\delta^{-1}\BR\Lambda_\delta}\\
&=\int\limits_{\R^{2\times 2}}
  \1_{\{\norm{X}<1 \,\&\, \norm{\Lambda_\delta^{-1}X\Lambda_\delta}<1\}} 
  \dint\la{4}(X)\\
&= 2\int\limits_{-\infty}^\infty\int\limits_0^{2\pi}\int\limits_0^\infty
  \int\limits_0^\infty\1_{r<\min   \gz{\frac{1}{\norm{Y_{+}(t,\rho,\phi)}},
  \frac{1}{\norm{Y_{+}(t-\delta,\rho,\phi)}}}} r^3 
  \dint r \dint\rho \dint\phi \dint t \\
&+2\int\limits_{-\infty}^\infty\int\limits_0^{2\pi}\int\limits_0^\infty
  \int\limits_0^\infty\1_{r<\min\gz{\frac{1}{\norm{Y_{-}(t,\rho,\phi)}},
  \frac{1}{\norm{Y_{-}(t-\delta,\rho,\phi)}}}} r^3
\dint r \dint\rho \dint\phi \dint t.
\end{align*} 
Note that $Y_{\pm}(t,\rho,\phi)\in\text{SL}_2(\R)$ and by
  Lemma \ref{lem:svratio}, we have 
\begin{equation*}
\norm{Y_\pm (t,\rho,\phi)}= 
\exp\gz{\frac{1}{2}\cosh^{-1}\gz{\frac{\hsnorm{Y_\pm (t,\rho,\phi)}^2}{2}}},
\end{equation*} 
  where
\begin{equation*}
\hsnorm{Y_\pm (t,\rho,\phi)}^2
= 2\gz{\frac{\rho}{2}+\abs{\frac{\rho}{2}\sin(2\phi)-1}\cosh(2t)}
\end{equation*}
  which means 
\begin{equation*}
\hsnorm{Y_\pm (t-\delta,\rho,\phi)}^2>\hsnorm{Y_\pm (t,\rho,\phi)}^2
  \quad\text{if and only if}\quad
\abs{t-\delta}>\abs{t}\Leftrightarrow t<\delta/2.
\end{equation*}
With this observation, the previous integral can be written as
\begin{align*}
&\int\limits_{-\infty}^\infty\int\limits_0^{\pi}\int\limits_0^\infty
  \int\limits_0^{\infty}\1\gz{r<e^{-\frac{1}{2}\cosh^{-1}
  \max\gz{\frac{\hsnorm{Y_\pm (t,\rho,\phi)}^2}{2},
  \frac{\hsnorm{Y_\pm (t-\delta,\rho,\phi)}^2}{2}}}}  4r^3
  \dint r \dint\rho \dint\phi \dint t\\
&=\int\limits_{-\infty}^\infty\int\limits_0^{\pi}\int\limits_0^\infty
e^{-2\cosh^{-1}
  \max\gz{\frac{\hsnorm{Y_\pm (t,\rho,\phi)}^2}{2},
  \frac{\hsnorm{Y_\pm (t-\delta,\rho,\phi)}^2}{2}}}
  \dint\rho \dint\phi \dint t \\
&=\int\limits_{-\infty}^\frac{\delta}{2}\int\limits_0^{\pi}\int\limits_0^\infty
  e^{-2\cosh^{-1}\gz{\frac{\hsnorm{Y_\pm (t-\delta,\rho,\phi)}^2}{2}}}
  \dint\rho \dint\phi \dint t 
  +\int\limits_{\frac{\delta}{2}}^\infty\int\limits_0^{\pi}\int\limits_0^\infty
e^{-2\cosh^{-1}\gz{\frac{\hsnorm{Y_\pm (t,\rho,\phi)}^2}{2}}}
  \dint\rho \dint\phi \dint t \\
&=\underbrace{
  \int\limits_{-\infty}^\infty\int\limits_0^{\pi}\int\limits_0^\infty
  e^{-2\cosh^{-1}\gz{\frac{\hsnorm{Y_\pm (t,\rho,\phi)}^2}{2}}}
  \dint\rho \dint\phi \dint t}_{\frac{1}{2}\Vol(\BR)}
  -\underbrace{
  \int\limits_{-\frac{\delta}{2}}^\frac{\delta}{2}\int\limits_0^{\pi}
  \int\limits_0^\infty 
  e^{-2\cosh^{-1}\gz{\frac{\hsnorm{Y_\pm(t,\rho,\phi)}^2}{2}}}
  \dint\rho \dint\phi \dint t}_{\frac{1}{2}\Delta (\delta)},
\end{align*}
  where the last term can be written as
\begin{align}
\begin{split}\label{eq:halfdelta}
\frac{1}{2}\Delta (\delta) 
  &=\int\limits_{-\frac{\delta}{2}}^\frac{\delta}{2}\int\limits_0^{\pi}
  \int\limits_0^\infty 
  e^{-2\cosh^{-1}\gz{\frac{\hsnorm{Y_\pm (t,\rho,\phi)}^2}{2}}}
  \dint\rho \dint\phi \dint t \\
&=\int\limits_{-\frac{\delta}{2}}^\frac{\delta}{2}\int\limits_0^{2\pi}
  \int\limits_0^\infty
  e^{-2\cosh^{-1}\gz{\rho+\abs{\rho\sin\phi-1}\cosh(2t)}}
  \dint\rho \dint\phi \dint t \\
&=\int\limits_{0}^\delta \int\limits_0^{2\pi}\int\limits_0^\infty
  e^{-2\cosh^{-1}\gz{\rho+\abs{\rho\sin\phi-1}\cosh t}}
  \dint\rho \dint\phi \dint t.
\end{split}
\end{align}
We decompose the inner double integral in the following way
\begin{align*}
\int\limits_0^{2\pi}\int\limits_0^\infty &
  e^{-2\cosh^{-1}\gz{\rho+\abs{\rho\sin\phi-1}\cosh t}}
  \dint\rho \dint\phi 
=\int\limits_0^{\pi}\int\limits_0^\infty 
  e^{-2\cosh^{-1}\gz{\rho+\gz{\rho\sin\phi+1}\cosh t}}\dint\rho \dint\phi \\
&+\int\limits_0^{\pi}\gz{\int\limits_0^\frac{1}{\sin\phi}
  e^{-2\cosh^{-1}\gz{\rho-\gz{\rho\sin\phi-1}\cosh t}}\dint\rho
  +\int\limits_{\frac{1}{\sin\phi}}^\infty
  e^{-2\cosh^{-1}\gz{\rho +\gz{\rho\sin\phi-1}\cosh t}}\dint\rho}\dint\phi.
\end{align*}
After some manipulation with the inner integrals we have
\begin{align*}
&\int\limits_0^{2\pi}\int\limits_0^\infty
  e^{-2\cosh^{-1}\gz{\rho+\abs{\rho\sin\phi-1}\cosh t}}\dint\rho \dint\phi
=\int\limits_0^{\pi}\frac{1}{1+\cosh t\sin\phi}\int\limits_t^\infty
  e^{-2u}\sinh u\dint u\dint\phi \\
&+\int\limits_0^{\pi}\gz{\frac{1}{1-\cosh t\sin\phi}
  \int\limits_t^{\cosh^{-1}\gz{\frac{1}{\sin\phi}}}e^{-2u}\sinh u\dint u
  +\frac{1}{1+\cosh t\sin\phi}
  \int\limits_{\cosh^{-1}\gz{\frac{1}{\sin\phi}}}^\infty e^{-2u}\sinh u\dint u}
  \dint\phi \\
&=2\int\limits_0^\frac{\pi}{2}\gz{e^{-t}-\frac{e^{-3t}}{3}
  -\gz{\tan\frac{\phi}{2}-\frac{1}{3}\tan^3\frac{\phi}{2}}\cosh t\sin\phi}
  \frac{1}{1-\cosh^2 t\sin^2\phi}\dint\phi,
\end{align*}
  where we applied the identity 
  $\exp\gz{-\cosh^{-1}\gz{\frac{1}{\sin\phi}}}=\tan\frac{\phi}{2}$. 
Now we substitute $\tan\frac{\phi}{2}=e^{-s}$ and we get
\begin{align*}
2\int\limits_0^\infty 
& \gz{e^{-t}-\frac{e^{-3t}}{3}-\gz{e^{-s}-\frac{e^{-3s}}{3}}
  \frac{\cosh t}{\cosh s}}
  \frac{1}{1-\gz{\frac{\cosh t}{\cosh s}}^2}\frac{1}{\cosh s}\dint s \\
&=\frac{8}{3}\int\limits_0^\infty 
  e^{-t}\cosh s-\frac{\sinh^2 s}{\sinh (t+s)}\dint s 
  =\frac{8}{3}\gz{\cosh t - \sinh^2 t\log \gz{\frac{e^t+1}{e^t-1}}}.
\end{align*}

For the defect function, we gain the following formula
\begin{equation}\label{defect}
\Delta (\delta)=\frac{16}{3} \int\limits_0^\delta
  \cosh t - \gz{\sinh^2 t}\log \gz{\frac{e^t+1}{e^t-1}}\dint t
\end{equation}
  which completes the proof.

\section{Proof of Lemma \ref{lem:eta}}\label{App:AppendixB}

We cover the manifold $\partial \BR$ with the following atlas
\begin{equation}
\mathcal{A}=\gz{
  \frac{Y_{\pm}(t,\rho,\phi)}{\norm{Y_{\pm}(t,\rho,\phi)}},
  \frac{Y_{\pm}(t,\rho,\phi)\sigma_3}{\norm{Y_{\pm}(t,\rho,\phi)\sigma_3}}},
\end{equation}
  where $Y_{\pm}(r,t,\rho,\phi)$ is given by \eqref{parR} and $\norm{\cdot}$ 
  denotes the usual operator norm. 
Direct computation of the volume form from this parametrization would be a
  cumbersome task even for computer algebra systems.

It is obvious that the metric tensor has the same form on every element of
  $\mathcal{A}$ hence it is enough to deal with the parametrization
\begin{equation*}
X(t,\rho,\phi)=\frac{Y(t,\rho,\phi)}{\norm{Y(t,\rho,\phi)}},
\end{equation*}
  where $Y(t,\rho,\phi):=Y_{+}(t,\rho,\phi)$. 
Recall the fact that $Y(t,\rho,\phi)\in\text{SL}_2(\R)$ and by 
  Lemma \ref{lem:svratio}, we have
\begin{equation}
X(t,\rho,\phi) = f(t,\rho,\phi) Y(t,\rho,\phi),
\end{equation} 
  where 
\begin{equation}
f(t,\rho,\phi)
  =\exp \gz{-\frac{1}{2}\cosh^{-1}\gz{\frac{\hsnorm{Y(t,\rho,\phi)}^2}{2}}}.
\end{equation}

The metric tensor $(g)$ corresponding to this parametrization can be written as
\begin{align*}
\frac{1}{f^2} g_{ij}
&=\frac{1}{f^2} \left\langle\partial_i X, \partial_j X\right\rangle 
  =\gz{\partial_i \log (f)}\gz{\partial_j\log (f)}\hsnorm{Y}^2 \\
&+\frac{1}{2}\gz{\gz{\partial_i \log (f)} \gz{\partial_j \hsnorm{Y}^2}
  +\gz{\partial_j \log (f)} \gz{\partial_i \hsnorm{Y}^2}}
  +\left\langle\partial_i Y, \partial_j Y\right\rangle,
\end{align*}
  where $\left\langle , \right\rangle$ denotes the usual Hilbert--Schmidt 
  scalar product. 
By the chain rule, the metric tensor can be written in the following 
  convenient form
\begin{equation}
g=f^2 \gz{G+\gz{\hsnorm{Y}^2\gz{h'\gz{\hsnorm{Y}^2}}^2
  +h'\gz{\hsnorm{Y}^2}}\times 
  \nabla \gz{\hsnorm{Y}^2}\nabla \gz{\hsnorm{Y}^2}^T},
\end{equation}
  where  $G_{ij}=\left\langle\partial_i Y, \partial_j Y\right\rangle$ 
  and $h(r)=-\frac{1}{2}\cosh^{-1}\gz{\frac{r}{2}}$.

According to the matrix determinant lemma, we have
\begin{equation*}
\det(g)=f^6\det(G)\times 
  \gz{1+\gz{\hsnorm{Y}^2\gz{h'\gz{\hsnorm{Y}^2}}^2+h'\gz{\hsnorm{Y}^2}}
  \nabla \gz{\hsnorm{Y}^2}^TG^{-1}\nabla \gz{\hsnorm{Y}^2}}
\end{equation*}
  where all the factors can be directly evaluated. 
We obtain the following nice form for the volume form
\begin{equation}
\sqrt{\det (g)}=f^4=\exp\gz{-2\cosh^{-1}\gz{\frac{\hsnorm{Y}^2}{2}}}.
\end{equation}

Using the notations introduced in Appendix \ref{App:AppendixA}, we can write
\begin{align*}
\tilde{\eta}_1 (e^{-\delta})
&=4\int\limits_{-\infty}^\infty\int\limits_0^{2\pi}\int\limits_0^\infty
\frac{e^{-2\cosh^{-1}\gz{\frac{\hsnorm{Y(t,\rho,\phi)}^2}{2}}}}
  {\Vol(\partial\BR)}
  \1_{\norm{Y(t-\delta,\rho,\phi)}<\norm{Y(t,\rho,\phi)}}
  \dint\rho\dint\phi\dint t \\
&=\frac{4}{\Vol(\partial\BR)} \int\limits_{\frac{\delta}{2}}^\infty\int\limits_0^{2\pi}\int\limits_0^\infty
  e^{-2\cosh^{-1}\gz{\frac{\hsnorm{Y(t,\rho,\phi)}^2}{2}}}
  \dint\rho\dint\phi\dint t \\
&=1-\frac{4}{\Vol(\partial\BR)} 
  \underbrace{ \int\limits_0^\delta\int\limits_0^{2\pi}\int\limits_0^\infty
  e^{-2\cosh^{-1}\gz{\rho+|\rho\sin\phi-1|\cosh t}}
  \dint\rho\dint\phi\dint t}_{
  \frac{1}{2}\Delta (\delta)\,\,\text{(See \eqref{eq:halfdelta}.)}},
\end{align*}
  where we applied the following identities
\begin{align*}
\Lambda_\delta^{-1} Y(t,\rho,\phi) \Lambda_\delta &= Y(t-\delta,\rho,\phi) \\
\norm{Y(t-\delta,\rho,\phi)}<\norm{Y(t,\rho,\phi)} &\Leftrightarrow t>\delta /2.
\end{align*}
So, we have
\begin{equation*}
\tilde{\eta}_1 (\varepsilon)=1-\frac{2\Vol(\BR)}
  {\Vol(\partial\BR)}\gz{1-\tilde{\chi}_1 (\varepsilon)}
\end{equation*}
  which implies $\tilde{\eta}_1 (\varepsilon)=\tilde{\chi}_1 (\varepsilon)$
  for $\varepsilon\in [0,1]$ because $\tilde{\eta}_1 (0)=\tilde{\chi}_1 (0)=0$
  and $\tilde{\eta}_1 (1)=\tilde{\chi}_1 (1)=1$.

\end{document}